\documentclass{amsart}

\newtheorem{definition}{Definition}
\newtheorem{proposition}{Proposition}

\newtheorem{theorem}{Theorem}

\newtheorem{remark}{Remark}

\usepackage{amssymb}

\begin{document}

\author{A.M. Ionescu, V. Slesar, M. Visinescu and G.E. V\^{\i}lcu}
\title[Transversal Killing and twistor spinors]{Transversal Killing and twistor spinors associated to the basic Dirac operators}
\date{\today }
\maketitle

\begin{abstract}
We study the interplay between basic Dirac operator and transversal Killing and twistor
spinors. In order to obtain results for general Riemannian foliations with
bundle-like metric we consider transversal Killing spinors that appear as
natural extension of the harmonic spinors associated with the basic Dirac
operator. In the case of foliations with basic-harmonic mean curvature it
turns out that this type of spinors coincide with the standard definition.
We obtain the corresponding version of classical results on closed Riemannian
manifold with spin structure, extending some previous results.
\end{abstract}

Keywords: \emph{Riemannian foliations; basic Dirac operator; transversal Killing spinors;
transversal twistor spinors}.\\

Mathematics Subject Classification (2000): 53C12; 53C27; 57R30.

\section{Introduction}

On differentiable closed Riemannian manifolds the interplay between the spectrum of associated natural
differential operators such as Dirac operator and its square, the Laplace operator\-- on a side, and
the geometry of the underlying manifold\-- on the other side, has already become a classical research subject.

If we consider a foliated structure on our manifold such that the metric tensor field
is bundle-like (i.e. the manifold can be locally described as a Riemannian submersion \cite{Re}),
then similar operators can be defined in this particular setting, standing as important tools for the
study of the transverse spectral geometry of the foliations \cite{To}. Furthermore,
on the foliated manifold can be considered the existence of a symplectic structure or a CR-submanifold (see e.g. \cite{BJCF, CP, Ko, Vi1, Vi2}) which is also known to interact with the canonical differential operators \cite{Nic}.

Concerning the Dirac-type operators, the\emph{\ transversal Dirac}
operator for Riemannian foliations was introduced in \cite{Gl-Ka}. In the
particular case of a Riemannian foliation with basic mean curvature form,
this operator is used to define the\emph{\ basic Dirac} operator, which is
a symmetric, essentially self-adjoint and transversally elliptic \cite{Gl-Ka}. As pointed out by
the authors, Dirac-type operators defined in this particular framework are relevant at least in $\mathbb{R}^4$,
when the Yang-Mills equations can be dimensionally reduced, yielding magnetic monopole equations.

In general the procedure of reducing a dynamical system to two or more
systems of lower dimensions with respect to foliations of the manifold
of which the dynamics take place has proved to be an useful tool in
order to relate different integrable systems together with their
associated symmetries. Sometimes the reverse of the reduction procedure
is used to investigate difficult dynamical systems. It is also possible to sort of unfold the initial dynamics by imbedding it in a larger
one which is easier to integrate and then projecting the solution back
to the initial manifold \cite{MSS}.

On the other side, concerning the relevance of foliated manifolds, among a number of possible applications
we mention the problem of incorporating physical chiral spinors in
four dimensions in the framework of $N=1, D=11$ model of supergravity
\cite{EW}. In this respect, a $G$-foliated model of the $7$-dimensional
internal manifold resulting from $N=1, D=11$ model of supergravity was
proposed in \cite{CLC}. The ground-state compactification of the
$11$-dimensional manifold $M_{11}$ into $M_7 \times M_4$ is the result
of a global submersion $M_{11}$ into $M_4$, giving rise to a foliation
of dimension $7$ and codimension $4$.

Now, regarding the spectral properties of a Dirac operator in the standard framework
of a closed, Riemannian manifold we refer to \cite{Ba-Fr,Fr,Hij}; in this setting
the so called Killing spinors are known to be related to the spectrum of Dirac operator in a particular way;
beside the fact that they are a tool which allow us to construct Killing vector fields
(i.e. the corresponding flow is represented by local isometries) the Killing spinors
are exactly the eigenspinors corresponding to the lower eigenvalue of Dirac operator
\cite{Ba-Fr,Hij}, so the geometrical conditions that insure the existence of Killing spinors
and the realization of the limiting case for the Dirac spectrum are the same. A generalization of the Killing
spinors is represented by twistor spinors \cite{Ba-Fr, Gin}.

Concerning the \emph{transverse Killing spinors}, for the first time they appear in \cite{Al-Gr-Iv}, when the specific equation is written using derivatives in the transverse directions with respect to a vector field associated to a parallel 1-form. Moreover, the spinors are parallel in the direction of the vector field. In the framework represented by Riemannian foliations with arbitrary dimension, the transverse Killing spinors were introduced under this name in \cite{Ju1}. They verify a differential equation similar to the classical case; in order to investigate their properties using Dirac type equations and Lichnerowicz type formulas it is convenient to consider, as above, spinors parallel in the leafwise directions, i.e. \emph{basic spinors} \cite{Ju2, Gl-Ka, Gin-Hab1}. Recently, for the setting of Riemannian foliations with one-dimensional leaves (the so called \emph{Riemannian flows}), in \cite{Gin-Hab2} the authors introduced a more flexible concept. They consider spinors that are not basic anymore, but the first derivatives of the spinor field can behave differently along the transverse and leafwise distributions. Regarding the \emph{transverse twistor spinors}, they were defined and studied in \cite{Ju2}.

The main goal of this paper is to obtain the corresponding interplay between transversal Killing spinors and basic Dirac operator in the framework of Riemannian foliations. In order to made such an achievement in the general setting of Riemannian foliations we define transversal Killing spinors as natural extension of basic spinors transversally parallel with respect to a modified connection associated with the basic Dirac operator (see Section \ref{geometric objects}). We also introduce in a natural way a class of twistor spinors.

For general Riemannian foliations these definitions differs from \cite{Gin-Hab1, Ju1, Ju2}, but for the particular case of Riemannian foliation with basic-harmonic curvature, which is the most convenient setting, this definition coincide with the previous definition \cite{Gin-Hab1,Ju1,Ju2}, and our results turn out to be generalization of \cite{Ju1,Ju2}. Moreover, in the standard manner \cite{Gl-Ka,Ju1}, the absolute case (when the manifold is foliated by points) becomes a generalization of the case of closed Riemannian manifolds.

In the second section we introduce the main geometrical object we are dealing with. In the third section we derive the main results and point out the specific features of the above Killing spinors, while in the fourth section of the paper we study twistor spinors in the presence of a foliated structure. In the final part of the paper we point out some possible applications of the results and some physical considerations.

\section{Geometric objects related to the transverse geometry of Riemannian
foliations}\label{geometric objects}

The framework of this paper is represented by a smooth, closed (i.e. compact, without boundary) Riemannian manifold $M$ with a foliated structure $\mathcal{F}$. We also consider a metric tensor $g$ which is bundle-like \cite{Re}. The \emph{leafwise} distribution tangent to leaves will be denoted by $T\mathcal{F}$;  a  corresponding \emph{transverse} distribution $Q=T\mathcal{F}^{\perp }\simeq TM/T\mathcal{F}$ is obtained. Let us assume $\dim M=n$, $\dim T\mathcal{F}=p$ and $\dim Q=q$, with $p+q=n$.

A first consequence is the splitting of the tangent and the cotangent vector bundles associated
with $M$
\begin{eqnarray*}
TM &=&Q\oplus T\mathcal{F}, \\
TM^{*} &=&Q^{*}\oplus T\mathcal{F}^{*}.
\end{eqnarray*}
The canonical projection operator associated with the distributions $Q$ will be denoted by $\pi _Q$.

For local investigation of our manifold we will use local vector fields $\left\{e_i,f_a\right\} $ defined on a neighborhood
of a point $x\in M$ inducing an orthonormal basis at any point where they are
defined, $\left\{ e_i\right\}_{1\le i \le q} $ spanning the distribution $Q$ and
$\left\{f_a\right\}_{1\le a \le p} $ spanning the distribution $T\mathcal{F}$.

A convenient tool for the study of the basic geometry of the Riemannian foliated manifold $\left(M, \mathcal{F},g\right)$, is the so called \emph{Bott connection} (see e.g. \cite{To}, which is a linear, metric and  torsion-free connection). If we denote by $\nabla ^g$ the canonical Levi-Civita connection, then on the transverse distribution $Q$
we can define the connection $\nabla $ by the following relations
\[
\left\{
\begin{tabular}{l}
$\nabla _UX:=\pi _Q\left( [U,X]\right)$, \\
$\nabla _YX:=\pi _Q\left( \nabla _Y^gX\right)$,
\end{tabular}
\right.
\]
for any smooth sections $U\in \Gamma \left( T\mathcal{F}\right)$, $X$, $Y\in
\Gamma \left( Q\right) $. In a standard manner we can associate to $\nabla $ the
\emph{transversal Ricci operator} $\mathrm{Ric}^\nabla $, the \emph{transversal scalar
curvature} $\mathrm{Scal}^\nabla $  and the \emph{transversal gradient} of a
basic function $f$ $\mathrm{grad}^{\nabla}$.

It is also convenient to employ basic (projectable) local vector field $\{e_i\}_{1\le i\le q}$, parallel on the leafwise
directions with respect to the above Bott connection, so we use this type of transverse orthonormal basis throughout the paper. We
denote by $\Gamma_b(Q)$ the set of basic vector fields.

The classical de Rham complex of differential forms $\Omega
\left( M\right) $ is restricted to the complex of \emph{basic differential forms}, defined as \cite{To}

\[
\Omega _b\left( M\right) :=\left\{ \omega \in \Omega \left( M\right) \mid
\iota _U\omega =\mathcal{L}_U\omega =0\right\} ,
\]
where $U$ is again an arbitrary leafwise vector field, $\mathcal{L}$ being
the Lie derivative along $U$, while $\iota $ stands for interior product.
The corresponding basic exterior derivative $d_b$ comes as a restriction of the classical de Rham derivative, $d_b:=d_{\mid \Omega_b\left( M\right) }$. Let us notice that basic de Rham complex is defined independent
of the metric structure $g$. The adjoint operator, namely the \emph{basic co-derivative} $\delta _b$, can also be considered (see e.g.
\cite{Al}).

One differential form of particular importance, which may not be necessarily a basic differential form, is the \emph{mean
curvature form}. In order to define it,  first of all we set
$k^{\sharp}:=\pi_Q\left(\sum_a\nabla^g_{f_a}f_a\right)$ to be the
\emph{mean curvature vector field} associated with the distribution
$T\mathcal{F}$. Then, $k$ will be the mean curvature form which is subject to
the condition $k(U)=\langle k^{\sharp},U \rangle$, for any vector field $U$,
$\sharp$ being the \emph{musical} isomorphism and $\left\langle \cdot ,\cdot \right\rangle $
the scalar product in $TM$.

The de Rham complex can be decomposed as a direct sum \cite[Theorem 2.1]{Al}
\begin{equation}\label{hodge decomposition}
\Omega \left(M\right) =\Omega_b\left( M\right) \bigoplus \Omega_b\left(
M\right)^{\perp},
\end{equation}
with respect to the $C^\infty -$Frechet topology. Consequently, on any Riemannian
foliation the mean curvature form can be written as
\[
k=k_b+k_o,
\]
where $k_b\in \Omega _b\left( M\right) $ is the $\emph{basic}$ component of
the mean curvature, $k_o$ being the orthogonal complement. In the following we
denote $\tau :=k_b^{\sharp }$.

\begin{remark}
The above co-derivative operator $\delta_b$ can be calculated using the
vector field $\tau$ and the Bott connection (see \cite{To, Al})
\[
\delta_b=-\sum_i \iota_{e_i} \nabla_{e_i} +\iota_{\tau}.
\]
\end{remark}

With the above notations, at any point $x$ on $M$ we consider the Clifford
algebra $Cl(Q_x)$ which, with respect to the orthonormal basis $\{e_i\}$ is
generated by $1$ and the vectors $\{e_i\}$ over the complex field, being
subject to the relations $e_i\cdot e_j+e_j\cdot e_i=-2\delta _{i,j}$, $1\leq
i,j\le q$, where dot stands for Clifford multiplication. The resulting
bundle $Cl(Q)$ of Clifford algebras will be called the \emph{Clifford bundle}
over $M$, associated with $Q$.

The additional assumptions for the foliation $\mathcal{F}$ is the transverse orientability and the existence
of a transverse spin structure. This means that there exists a principal
${\rm Spin}(q)$-bundle $\tilde P$ which is a double sheeted covering of the
transverse principal $SO(q)$-bundle of oriented orthonormal frames $P$,
such that the restriction to each fiber induces the covering projection
${\rm Spin}(q)\rightarrow SO(q)$; such a foliation is called \emph{spin foliation}
(see e.g. \cite{Hab-Ric}). If $\Delta _q$ is the spin irreducible representation associated with
$Q$, then the \emph{foliated spinor bundle} $S:=\tilde
P\times _{{\rm Spin}(q)}\Delta _q$ can be constructed \cite{Roe}. Furthermore, in a classical way a smooth bundle action can be considered
\[
\Gamma \left( Cl(Q)\right) \otimes \Gamma \left( S\right) \longrightarrow
\Gamma \left( S\right) .
\]
We denote this action also with Clifford multiplication; it verifies the condition
\[
\left( u\cdot v\right) \cdot s=u\cdot \left( v\cdot s\right) ,
\]
for $u,v\in \Gamma \left( Cl(Q)\right) $, $s\in \Gamma \left( S\right) $.

It is easily seen that $S$ becomes a bundle of Clifford modules \cite{Roe}.
\begin{remark}
The above transverse Clifford action is obtained from the standard case of the
tangent bundle as a restriction of the Clifford action from $TM$
to the distribution $Q$ (for the particular case of vector bundles see also \cite{Va}).
\end{remark}

The lifting of the Riemannian connection on $P$ can be used to introduce
canonically a connection on $S$, which will be denoted also by $\nabla $.
The \emph{compatibility} between the Clifford action and the connection
$\nabla $ is expressed in the relation
\[
\nabla _U\left( u\cdot s\right) =\left( \nabla _Uu\right) \cdot s+u\cdot
\nabla _Us,
\]
for any $U\in \Gamma (TM)$, $u\in \Gamma \left( Cl(Q)\right) $,\,
$s\in \Gamma (S)$, extending canonically the connection $\nabla $ to $\Gamma
\left( Cl(Q)\right) $.

The transverse metric induces a hermitian structure on $S$ ; if
we denote it by $\left( \cdot \mid \cdot \right) $, we have that
$\left(X\cdot s_1\mid s_2\right) =-\left( s_1\mid X\cdot s_2\right) $, for any
$X\in \Gamma \left( Q\right) $, $s_1$, $s_2\in \Gamma \left( S\right) $. Then $\nabla$ becomes
a metric connection, similar to the classical case of spin manifolds (see e.g. \cite[Chapter 1]{Ba-Fr}) and we have
\[
X\left(s_1\mid s_2\right)=\left(\nabla_X  s_1\mid s_2\right)+\left(s_1\mid \nabla_X s_2\right).
\]

In order to define the \emph{basic Dirac operator}, we need first to introduce  the \emph{transversal
Dirac operator},
\[
D_{tr}:=\sum\limits_ie_i\cdot \nabla _{e_i},
\]
and the \emph{basic spinors} or holonomy invariant sections, in accordance with \cite{Gl-Ka}

\begin{equation}\label{basic_spinor}
\Gamma _b\left( S\right) :=\left\{ s\in \Gamma \left( S\right) \mid
\nabla_Us=0,\,\mbox{for any}\,U\in \Gamma \left( T\mathcal{F}\right) \right\}.
\end{equation}
The transversal Dirac operator may not be formally self-adjoint; consequently, in the definition of the \emph{basic Dirac operator} an auxiliary term related to the basic component of the mean curvature form is added
\begin{equation}
D_b:=\sum\limits_ie_i\cdot \nabla _{e_i}-\frac 12\tau ,  \label{basic Dirac}
\end{equation}
and the domain of this operator  is restricted to the above set of basic spinors \cite{Hab-Ric, Gl-Ka}.
\begin{remark}
As in the standard setting, the above basic Dirac operator does not depend on the local framework $\left\{e_i \right\}_{1\le i\le q}$. It is a transversally elliptic and essentially self-adjoint differential operator with respect to the inner product canonically associated with the hermitian structure. We emphasize the fact that the spectrum $\sigma (D_b)$ is discrete \cite{Gl-Ka}.
\end{remark}

Another way to construct the basic Dirac operator will be described in the
following. We start out by modifying the Bott connection.

\begin{definition}
The \emph{modified connection} on the space of basic sections $\Gamma_b\left( S\right) $ is given by \cite{Sl}

\begin{equation}
\bar \nabla _Xs:=\nabla _Xs-\frac 12\left\langle X,\tau \right\rangle s,
\label{modified connection}
\end{equation}
for any $X\in \Gamma\left( TM \right) $ and $s\in \Gamma \left( S\right) $.
\end{definition}

A key property of $\bar \nabla$ is that it can be used to construct the basic Dirac operator \cite{Ko}:
\begin{eqnarray}
\sum_ie_i\cdot \bar \nabla _{e_i} &=&\sum_ie_i\cdot \left( \nabla_{e_i}-\frac 12\left\langle e_i,\tau \right\rangle \right)  \nonumber \\
\ &=&\sum_ae_i\cdot \nabla _{e_i}-\frac 12\tau  \nonumber \\
\ &=&D_b.  \nonumber
\end{eqnarray}

\begin{remark}
It is easy to see that the modified connection on $S$ is also compatible
with the Bott connection, i.e.
\[
\bar \nabla _X\left( Y\cdot s\right) =\nabla _XY\cdot s+Y\cdot \bar \nabla _Xs,
\]
for any $X$, $Y\in \Gamma (Q)$. On the other side we must observe that the modified
connection does not share other classical properties. Unlike the canonically connection $\nabla$, the modified connection is not a metric connection, and in general this aspect has impact on the formal computation.
\end{remark}

Now, as we have already defined the basic Dirac operator, we introduce a
category of spinors intimately related to this differential operator. In
order to consolidate the motivation let us study the following example.

We consider the torus $T^2:=R^2/Z^2$ with the metric $g=$ $e^{2f(y)}dx^2+dy^2
$, for some periodic function $f$ (see e.g. \cite{Hab-Ric}). As a
consequence, $\left\{ \partial _y,\frac{\partial _x}{e^{f(y)}}\right\} $
will be an orthonormal basis at any point, $Q=span\{\partial _y\}$,
$T\mathcal{F}$ $=span\{\partial _x/e^{f(y)}\}$. The spin structure considered on the transverse circles will be the
trivial one (see e. g. \cite{Gin}); the transverse Clifford multiplication by $\partial _y$ will be represented by
multiplication with the purely imaginary unit $i$. We use the Koszul formula to calculate the mean curvature vector field and
the mean curvature form.
\begin{eqnarray*}
\left\langle \nabla_{\frac{\partial _x}{e^{f(y)}}}^g\frac{\partial _x}{e^{f(y)}},\partial y\right\rangle  &=&\frac 12\left( \left\langle \left[
\partial y,\frac{\partial_x}{e^{f(y)}}\right] ,\frac{\partial_x}{e^{f(y)}}
\right\rangle +\left\langle \left[ \partial y,\frac{\partial_x}{e^{f(y)}}
\right] ,\frac{\partial_x}{e^{f(y)}}\right\rangle \right)  \\
&=&-f^{\prime}(y).
\end{eqnarray*}
Consequently, we get $k=-f^{\prime }(y)dy$; the mean curvature form is obviously basic, as it does not depend on $x$, so $k\equiv k_b$ and we obtain
\[
\tau =-f^{\prime }(y)\partial _y.
\]

On the other side it is interesting to note that
\begin{eqnarray*}
\delta _bk &=&\left( -\iota _{\partial y}\nabla _{\partial y}
+\iota_{-f^{\prime }(y)\partial _y}\right) \left( -f^{\prime }(y)dy\right)  \\
&=&f^{\prime \prime }(y)+\left( f^{\prime }(y)\right) ^2,
\end{eqnarray*}
so $\delta _bk$ does not necessarily vanish, i.e. $k$ is not a basic-harmonic
differential 1-form.

Finally, we calculate the basic Dirac operator
\begin{eqnarray*}
D_b &=&i\partial _y-\frac 12\left( -f^{\prime }(y)\right) i \\
&=&i\left( \partial _y+\frac 12f^{\prime }(y)\right) .
\end{eqnarray*}

We investigate the harmonic basic spinors. It is easy to see that the basic solutions of the equation $D_bs_1=0$ have the form $s_1=ce^{-\frac 12f(y)}$,
$c\in \mathbb{C}$.

\begin{remark}
From the above calculations we see that $\bar \nabla _{\partial _y}s_1=0$
and $\nabla _{\partial _y}s_1\not =0$. So, even for the above simple example
of Riemannian foliation, we see that the harmonic spinors of the basic Dirac
operator are parallel spinors with respect to the above modified connection $\bar \nabla$
but not with respect to the classical connection $\nabla$; also, the above spinors
are not transverse Killing spinors with the definitions from \cite{Ju1,Ju2, Gin-Hab1,Gin-Hab2}.
\end{remark}

In the classical setting, a category of spinors naturally related to parallel
spinors and eigenspinors  associated to the lower eigenvalue is represented by Killing spinors. In our
particular framework we introduce a similar type of spinors with respect to
our connection $\bar \nabla $.

\begin{definition}
A spinor $s\in \Gamma_b (S)$ which satisfies the equation
\[
\bar \nabla _Xs+\frac \lambda qX\cdot s=0,
\]
for any $X\in \Gamma_b (Q)$ is called \emph{transversal Killing spinor
associated with the connection }$\bar \nabla $ or $\tau -$\emph{Killing
spinor}.
\end{definition}

\begin{remark}
As in the standard setting, it is easy to see that a basic spinor parallel with respect to $\bar{\nabla}$ is a $\tau-$ Killing spinor. Also, each $\tau-$ Killing spinor is an eigenspinor of the basic Dirac operator.
\end{remark}

The concept of Killing spinors can be extended to \emph{twistor spinors} in a classical manner \cite{Ba-Fr,Gin}.

\begin{definition}\label{twistor spinor}
We denote by \emph{$\tau -$ twistor spinors (twistor spinors with respect to the modified connection)} a
basic spinor satisfying the following equation
\begin{equation}
\bar \nabla _Xs+\frac 1qX\cdot D_bs=0.  \label{twist_eq}
\end{equation}
\end{definition}

\begin{remark}
We will see in the next sections that these definitions
extend previous basic spinors existing in the particular case of Riemannian foliations
with basic-harmonic mean curvature \cite{Gin-Hab1, Ju1, Ju2}. As the linear connection employed
to define $\tau$-Killing spinors in not metric, it will  be interesting to point out that
they do not have constant length, unlike the previous definition.
\end{remark}

Now, as we have defined all necessary transverse geometric objects, in the
final part of this section we shortly present some results that we employ in order
to study spectral properties of basic Dirac operators. It turns out that,
the most convenient setting is represented by the case of Riemannian
foliations with basic-harmonic mean curvature, that is $d_bk=0$, $\delta_bk=0$.
Then, in order to obtain results in the general case, we need a
sequence of metric changes that leave the transverse metric on the normal
bundle intact.

A first relevant result in this direction is due to Dom\'inguez.

\begin{theorem}\cite{Do}
The bundle-like metric can be transformed (leaving the transverse metric on
the normal bundle intact) such that the orthogonal part $k_o$ of the mean curvature vanishes
while the basic part of the mean curvature $k_b$ holds; consequently the new bundle-like metric has basic
mean curvature.
\end{theorem}

The above metric transformation is based on a change of the transversal
sub-bundle $Q$, and a conformal change of the leafwise metric (see the proof of \cite[Theorem 9.18]{Do}).
These are in fact the fundamental metric changes needed in order to study the basic
component of the mean curvature \cite{Al}.

Secondly, we have the following result obtained by Mason.

\begin{theorem}\cite{Mas}
Furthermore, the above bundle-like metric can be transformed (leaving the
transverse metric on the normal bundle intact) into a metric with
basic-harmonic mean curvature.
\end{theorem}

This metric change is in fact a conformal change of the leafwise metric
which use the theory of stochastic flows.

Finally, an important spectral rigidity result is due to Habib and
Richardson.

\begin{theorem}\cite{Hab-Ric}
The spectrum $\sigma (D_b)$ is invariant with respect to any metric change that
leaves the transverse metric on the normal bundle intact.
\end{theorem}

As a consequence, we can study first of all the spectrum of basic Dirac operator
in the most convenient framework of Riemannian foliations with basic-harmonic mean curvature,
then pull back the results in the initial general case using the spectral rigidity result \cite{Hab-Ric}.
For the case $q \ge 2$ the authors obtain that the lower bound estimate for an eigenvalue $\lambda$ of $D_b$
\[
\lambda^2 \ge \frac{1}{4} \frac{q}{q-1}\mathrm{Scal}^{\nabla}_0,
\]
where $\mathrm{Scal}^{\nabla}_0:=\min_{x\in M} \mathrm{Scal}^{\nabla}_x$, as an extension of \cite{Ju1}.

While the above method is very useful when studying the eigenvalues of
$D_b$, in general the corresponding eigenspinors are not invariant with respect to all the above changes of the metric.

\section{Basic Dirac operators and transversal $\tau -$Killing spinors}

In the framework of closed Riemannian manifold with spin structure, the
interplay between Dirac operator and Killing spinors provide some remarkable
and very interesting results \cite{Ba-Fr,Gin}. In this section we study this
interplay between basic Dirac operator and transversal $\tau -$Killing
spinors on Riemannian foliations, obtaining corresponding results in our
specific setting, generalizing also known results from basic-harmonic
Riemannian foliations \cite{Ju1,Ju2}.
We start out by considering a Riemannian foliation with a bundle-like metric, with an arbitrary, non-necessary basic-harmonic mean curvature form, defined on a closed manifold.

\begin{definition}
For any real valued basic function $f$ we define on $S$ the linear connections $\nabla _X^f$,
and $\bar \nabla _X^f$,
\begin{eqnarray*}
\nabla_X^f:=\nabla _X+fX\cdot, \\
\bar \nabla_X^f:=\bar \nabla_X+fX\cdot.
\end{eqnarray*}
for any $X\in \Gamma_b(Q)$.
\end{definition}

\begin{proposition}
In the above setting, considering the corresponding curvature operators, we have the equality
$\bar R_{X,Y}^f=R_{X,Y}^f$ for $X$, $Y \in \Gamma_b(Q)$.
\end{proposition}
\begin{proof}Using the standard definition of the curvature operator, we start
with the relation
\begin{equation}
\bar R_{X,Y}^f=\bar \nabla _X^f\bar \nabla _Y^f-\bar \nabla _Y^f\bar \nabla_X^f-\bar \nabla _{[X,Y]}^f.  \label{bar_curvature}
\end{equation}
Furthermore, we get
\begin{eqnarray*}
\bar \nabla _X^f\bar \nabla _Y^f &=&\left( \nabla _X^f-\frac 12\left\langle
X,\tau \right\rangle \right) \left( \nabla _Y^f-\frac 12\left\langle Y,\tau
\right\rangle \right)  \\
&=&\nabla _X^f\nabla _Y^f-\frac 12\left\langle \nabla _XY,\tau \right\rangle
-\frac 12\left\langle Y,\nabla _X\tau \right\rangle  \\
&&-\frac 12\left\langle Y,\tau \right\rangle \nabla _X-\frac 12f\left\langle
Y,\tau \right\rangle X-\frac 12\left\langle X,\tau \right\rangle \nabla _Y \\
&&-\frac 12f\left\langle X,\tau \right\rangle Y+\frac 14\left\langle X,\tau
\right\rangle \left\langle Y,\tau \right\rangle .
\end{eqnarray*}

We also have the corresponding relation for the second term of (\ref{bar_curvature}); for the third term we have
\[
\bar \nabla _{[X,Y]}^f=\nabla _{[X,Y]}^f-\frac 12\left\langle [X,Y],\tau
\right\rangle .
\]

Now, let us emphasize that $k_b$ is a closed 1-form \cite[Corollary 3.5]{Al},
so
\[
\left\langle Y,\nabla _X\tau \right\rangle =\left\langle X,\nabla _Y\tau
\right\rangle .
\]
Also, the Bott connection is torsion-free, so we get
\[
\left\langle \nabla _XY-\nabla _XY-[X,Y],\tau \right\rangle =0.
\]

Summing up, we observe that all terms containing $\tau $ vanish and we obtain
the above relation.
\end{proof}

\begin{remark}
For $f\equiv 0$ we derive that $\bar R_{X,Y}=R_{X,Y}$.
\end{remark}

Now let us study the spin curvature operator applied on $\Gamma_b (S)$. A Riemannian foliation can be locally identified
with a Riemannian submersion \cite{Re}, and we can consider a local \emph{transversal}, i.e. a local transverse base manifold $T$,
so all the geometric objects we use are basic (projectable to the local transversal), and we can locally
identify the transverse spin curvature operator with the spin curvature
operator on the transverse manifold $T$. Then, similar to the classical case \cite[Equation 1.13]{Ba-Fr}, we get the following corresponding local identity for the framework of Riemannian foliations \cite[Equation (4.4)]{Ju1}, which will  be useful in our further considerations.

\begin{equation}
\sum_i e_i\cdot \mathrm{Ric}^{\nabla}(e_i)\cdot s = -\mathrm{Scal}^\nabla s \label{Ricci spin}.
\end{equation}

We also have
\begin{equation}
\sum_i e_i\cdot R_{X,e_i}s =-\frac 12 \mathrm{Ric}^{\nabla}(X)\cdot s, \label{scalar spin}.\\
\end{equation}
In the following let us denote the spectrum of the basic Dirac operator defined on a Riemannian foliations with $\sigma (D_b)=\{\lambda_k\}_{k \ge 1}$, such that $|\lambda_1|\le |\lambda_2| \le ...$ .
\begin{proposition}\label{jj}
On a closed Riemannian foliation of codimension $q \ge 2$ endowed with a transverse spin structure we assume the existence of a basic spinor $s_1\in \Gamma_b (S)$ which verifies the equation
\begin{equation}\label{hijazi}
\bar \nabla _Xs_1+\frac fqX\cdot s_1=0,
\end{equation}
for any $X\in \Gamma_b (Q)$, $f$ being a basic real-valued function. Then $f$ is constant
$f=\pm \lambda _1$, (as a consequence the spinor is transverse $\tau -$Killing), the foliation is transversally Einstein, the transversal scalar curvature $\mathrm{Scal}^\nabla $ is a positive constant function $\mathrm{Scal}^\nabla \equiv
Scal_0^\nabla >0$, and
\[
\lambda^2 _1=\frac 14\frac q{(q-1)} \mathrm{Scal}^{\nabla}_0.
\]
\end{proposition}

\begin{proof}From (\ref{hijazi}) we have
\begin{eqnarray*}
0 &=&\sum_ie_i\cdot \bar R_{X,e_i}^{\frac{f}{q}}s_1 \\
\  &=&\sum_ie_i\cdot R_{X,e_i}^{\frac{f}{q}}s_1.
\end{eqnarray*}

On the other side
\[
R_{X,e_i}^{\frac fq}s_1=\nabla _X^{\frac fq}\nabla _{e_i}^{\frac fq}
-\nabla_{e_i}^{\frac fq}\nabla _X^{\frac fq}-\nabla _{\pi _Q([X,e_i])}^{\frac fq}s_1.
\]
As $s\in \Gamma _b\left( S\right) $ is a basic spinor parallel in the
leafwise direction, we can adjust accordingly the last term in the
expression of curvature operator. Consequently
\begin{eqnarray*}
R_{X,e_i}^{\frac fq}s_1 &=&R_{X,e_i}s_1 \\
&&+\frac fqX\cdot \nabla_{e_i}s_1+X\left( \frac fq\right) e_i\cdot s_1+\frac fq\nabla _Xe_i\cdot
s_1+\frac fqe_i\cdot \nabla _Xs_1+\frac{f^2}{q^2}X\cdot e_i\cdot s_1 \\
&&-\frac fqe_i\cdot \nabla _Xs_1-e_i\left( \frac fq\right) X\cdot s_1-\frac
fq\nabla _{e_i}X\cdot s_1-\frac fqX\cdot \nabla _{e_i}s_1-\frac{f^2}{q^2}
e_i\cdot X\cdot s_1 \\
&&-\frac fq\pi _Q([X,e_i])\cdot s_1 \\
&=&R_{X,e_i}s_1+X\left( \frac fq\right) e_i\cdot s_1-e_i\left( \frac
fq\right) X\cdot s_1+\frac{f^2}{q^2}X\cdot e_i\cdot s_1-\frac{f^2}{q^2}
e_i\cdot X\cdot s_1 \\
&&+\frac fq\left( \nabla _Xe_i-\nabla _{e_i}X-\pi _Q([X,e_i])\right) \cdot
s_1.
\end{eqnarray*}

We notice that the last term vanishes, being the torsion of the Bott connection
\cite[Proposition 3.8]{To}. Then
\begin{eqnarray*}
\sum_ie_i\cdot R_{X,e_i}^{\frac fq}s_1 &=&\sum_ie_i\cdot
R_{X,e_i}s_1-qX\left( \frac fq\right) s_1-\sum_ie_i\left( \frac fq\right)
e_i\cdot X\cdot s_1 \\
&&+\frac{f^2}{q^2}\sum_ie_i\cdot X\cdot e_i\cdot s_1+q\frac{f^2}{q^2}X\cdot
s_1.
\end{eqnarray*}

Let us notice that
\[
\sum_ie_i\left( \frac fq\right) e_i=\mathrm{grad}^\nabla \left( \frac fq\right) ,
\]
and
\begin{eqnarray*}
\sum_ie_i\cdot X\cdot e_i &=&\sum_i-2\left\langle e_i,X\right\rangle
e_i-\sum_iX\cdot e_i\cdot e_i \\
&=&-2X+qX \\
&=&\left( q-2\right) X.
\end{eqnarray*}
Finally, using also \ref{Ricci spin}, we end up with the relation \cite{Ju1} (for the classical framework of a closed
Riemannian manifold see \cite[Equation 3.5]{Hij})
\[
-\frac 12\sum_iR_{X,e_i}e_i\cdot s_1-qX(\frac fq)s_1- \mathrm{grad}^{\nabla} (\frac fq)X\cdot
s_1+\frac {f^2}{q^2}(q-1)X\cdot s_1=0,
\]
Now, the conclusion comes from \cite{Ju1}.
\end{proof}

So, as a remark, the above fundamental properties of Killing spinor still
hold for our particular definition.

We present in the following the main result of our paper.
\begin{theorem}
Let us consider on a closed manifold a Riemannian foliation with a transverse spin structure; assume also that the mean curvature form is non-necessarily basic. Then the lower bound of the eigenvalues of the basic Dirac operator is attained
\[
\lambda_1^2=\frac 14\frac q{q-1}\mathrm{Scal}^{\nabla}_0,
\]
if and only if there exist a transverse $\tau -$ Killing spinor associated to one of the real numbers $+\frac 1q \sqrt{\frac 14\frac q{q-1}\mathrm{Scal}^{\nabla}_0}$ or $-\frac 1q \sqrt{\frac 14\frac q{q-1}\mathrm{Scal}^{\nabla}_0}$.
\end{theorem}

\begin{proof}: If we assume the existence of a transversal Killing spinor $s_1$
associated with $\bar \nabla $, then
\begin{eqnarray*}
D_bs_1 &=&\sum_ie_i\cdot \bar \nabla _{e_i}s_1 \\
&=&\mp\left( \sum_ie_i\cdot \frac 1q \sqrt{\frac 14\frac q{q-1}\mathrm{Scal}^{\nabla}_0} \,e_i\right) s_1 \\
&=&\pm \sqrt{\frac 14\frac q{q-1}\mathrm{Scal}^{\nabla}_0} \,s_1,
\end{eqnarray*}
so the lower bound estimate is attained, in accordance with \cite{Ju1,Hab-Ric}.

For the converse statement, as we use a Lichnerowicz type formula, let us observe that the basic
mean curvature vector field $\tau$ that appears in the definition of the basic Dirac operator is actually
obtained via the non-trivial Hodge type decomposition (\ref{hodge decomposition}) (see also \cite[Theorem 2.1]{Al}).
Then, for arbitrary Riemannian foliations is difficult to express $\tau$ and use it
in the standard calculation in order to get a Lichnerowicz formula (for an example a little bit more complicated than we presented in Section 2 see \cite[Example 2]{Hab-Ric}; compare also with the case of transverse Dirac operator,
for instance \cite[Theorem 4]{Ko}). In turn, we change the metric into a new bundle-like metric with basic mean curvature,
we obtain the result in this particular setting, then we pull-back the result in the general case, basically using the method from \cite{Sl}.
Consequently, let us consider a deformation of the metric as in
\cite{Do}. If the lower bound estimate of the spectrum is realized for the initial metric,
using the spectral rigidity result \cite{Hab-Ric}, the lower bound
estimate will be attainted also for the case of basic mean curvature form. In
this particular setting $\tau \equiv k^{\sharp}$ and we have the Lichnerowicz type formula \cite{Sl}

\begin{equation}
\left\| D_b s_1\right\| ^2=\left\| \bar \nabla s_1\right\| ^2+ \frac 14
\int\limits_M \mathrm{Scal}^{\nabla}\left| s_1 \right|^2,
\label{new_weitz_lich_integral_spin}
\end{equation}
for $s_1 \in \Gamma_b (S)$, where $\left| s_1\right| ^2=\left( s \mid s_1 \right)$, $\left\| \cdot \right\|$ being the $L^2$ norm associated with the hermitian structure.

In what follows, let us consider the connection
\[
\bar \nabla _X^{\frac{\lambda _1}q}:=\bar \nabla _X+\frac{\lambda _1}q X\cdot.
\]
when $\lambda_1$ is the eigenvalue for which the lower bound is attained, $s_1$ being the corresponding eigenspinor.
In the calculations below we follow closely in our particular setting the main steps from the original approach of Friedrich \cite{Fr} in order to get the lower bound estimate of the Dirac spectrum. Comparing with the corresponding approach of Jung
on Riemannian foliations \cite[Theorem 4.2]{Ju1}, as we use a different Lichnerowicz formula, let us notice the vanishing of the mean curvature term in the final estimate. Using then standard arguments, in the limiting case we
obtain necessary conditions for the eigenspinors.

Using (\ref{new_weitz_lich_integral_spin}), after calculations we get \cite{Sl}
\[
\int\limits_M\left| D_bs_1-\frac{\lambda _1}qs_1\right|
^2=\int\limits_M\sum_i\left| \bar \nabla _{e_i}^{\frac{\lambda _1}q}s_1\right|
^2+\frac 14\int\limits_M \mathrm{Scal}^\nabla \left| s_1\right| ^2+\int\limits_M\left(
1-q\right) \frac{\lambda _1^2}{q^2}\left| s_1\right| ^2,
\]
and, consequently \cite{Sl}
\[
\int\limits_M\left( \frac{q-1}q\lambda _1^2-\frac 14Scal^{\nabla}\right) \left|
s_1\right| ^2=\int\limits_M\sum_i\left| \bar \nabla _{e_i}^{\frac{\lambda _1}q}s_1\right|^2.
\]
From here, it turns out that
\[
\int\limits_M\sum_i\left| \bar \nabla _{e_i}^{\frac{\lambda _1}q}s_1\right|^2=0\,\,\mbox{and}\,\, \lambda _1^2=\frac 14\frac q{q-1}\mathrm{Scal}_x^\nabla,
\]
so
\[
\bar \nabla _Xs_1+\frac{\lambda _1}qX\cdot s_1=0,
\]
for any $X\in \Gamma \left( Q\right) $, i.e. the spinor $s$ is a transversal
$\tau -$Killing spinor.

As we noticed above, the metric change described in \cite{Do}
leaves the transverse metric and the basic part $k_b$ of the mean curvature
intact, so the action of the modified connection $\bar \nabla $ on
$\Gamma_b\left( S\right) $ does not change. On the other side the Clifford multiplications,
related to the transverse metric, always agree. As a consequence, the spinor $s$ will
be a transversal Killing spinor associated with $\bar \nabla $ for the
initial metric tensor, and we extend our result to arbitrary Riemannian
foliations using \cite{Do}.
\end{proof}

\begin{proposition}
If a Riemannian foliation defined on a closed manifold with $q \ge 2$ admits a transversal Killing spinor with respect to
the connection $\bar \nabla $, then the foliation is taut.
\end{proposition}

\begin{proof}Assuming that the Riemannian foliations admit a transversal Killing
spinor with respect to the connection $\bar \nabla $, we obtain that there
is an eigenspinor $s_1$ for which the lower bound estimate of the spectrum
of the basic Dirac operator is attained; from the Proposition \ref{jj} we also
obtain that in the case $q\ge 2$ the foliation should be transversally Einstein, and the
transversal scalar curvature must be non-negative and constant.
Using now once again the rigidity of the spectrum $\sigma (D_b)$, we obtain that
using a sequence of metric changes that hold the transverse part of the
metric we end up with a Riemannian foliation with a basic-harmonic mean
curvature for which the above lower bound is also realized and the square of the first
eigenvalue is again $\frac 14\frac q{q-1}\mathrm{Scal}^{\nabla}_0$. As $q \ge 2$, from \cite[Theorem 5.2]{Ju1} we get
that with respect to the deformed metric the foliation should be minimal,
i.e. $k\equiv 0$.
As a consequence, by the above sequences of metric changes
we obtain a Riemannian foliation with minimal leaves; as this is the
standard definition for a taut foliation (see e. g. \cite{Al}), the conclusion follows.
\end{proof}

\begin{remark}
For foliations with basic-harmonic mean curvature the existence of a
classical transversal Killing spinor is restricted by the condition $k\equiv 0$ \cite{Ju1}; as a
result $\bar \nabla \equiv \nabla $, and a spinors is transversally Killing
with respect to $\bar \nabla $ and $\nabla $ in the same time. Consequently,
the above results are natural generalization of \cite[Theorem 5.3]{Ju1} for basic spinors
and \cite[Corollary 4.5]{Ju2}.

\end{remark}

\section{Transversal $\tau -$ twistor spinors}

Within this section, in the same framework of Riemannian foliations with non-necessarily basic-harmonic mean curvature, using our previous method we study the main properties of the transversal $\tau -$ twistor spinors introduced in Section 2.

First of all we see that this concept is in fact an extension of \cite{Ju2}; there the twistor spinors are defined in the particular case of Riemannian foliations with basic mean curvature as basic spinors satisfying the classical equation written using the connection $\nabla$ on the spinor bundle $S$ and the basic Dirac operator
\[
\nabla_X s +\frac 1q X \cdot D_b s=0
\]
for any $X \in \Gamma_b(Q)$. Indeed, by \cite[Theorem 3.2]{Ju2}, the twistor spinors exist only on minimal foliations. As for minimal foliations $\tau =0$ and $\nabla \equiv \bar{\nabla}$, we see that our Definition \ref{twistor spinor} agrees with \cite{Ju2} for this particular class of Riemannian foliations. As the parallel spinors with respect to the connection $\bar{\nabla}$ considered in the Section \ref{geometric objects} are obviously $\tau-$ twistor spinor defined on a taut, non-minimal foliation,
the Definition \ref{twistor spinor} is in fact an extension of \cite{Ju2}; however, the fact that the $\tau -$twistor spinors exist only on taut foliation (as well as $\tau-$Killing spinors) does not seem to be a direct consequence.

In the calculations below let us assume $q>2$. We use the relations (\ref{scalar spin}), (\ref{Ricci spin}), the fact that Bott
connection has vanishing torsion and is compatible with the modified connection $\bar{\nabla}$, as well as the fact that the basic spinors are parallel on the leafwise directions.

\begin{eqnarray*}
\frac 12 \mathrm{Ric}^{\nabla}(X)\cdot s &=&-\sum_ie_i\cdot R_{X,e_j}s \\
&=&-\sum_ie_i\cdot \bar R_{X,e_j}s \\
&=&-\sum_ie_i\cdot \left( \bar \nabla _X\bar \nabla _{e_i}-\bar \nabla_{e_i}\bar \nabla _X-\bar \nabla _{\pi _Q([X,e_i])}\right) s. \\
&=&-\sum_ie_i \left( \cdot \bar \nabla _X\left( -\frac 1qe_i\cdot D_bs\right) -\bar
\nabla_{e_i}\left( -\frac 1qX\cdot D_bs\right) \right. \\
&& \left. -\frac 1q \pi_Q([X,e_i])\cdot D_bs \right ) \\
&=&-\sum_ie_i\cdot \left( -\frac 1qe_i\cdot \bar \nabla _XD_bs+\frac 1q X \cdot \bar \nabla_{e_i} D_bs\right)  \\
&&+\frac 1q\sum_ie_i\cdot \left( \nabla _Xe_i-\nabla _{e_i}X-\pi_Q([X,e_i])\right) \cdot D_bs.
\end{eqnarray*}

Now, as the transverse part of the torsion tensor associated to the Bott
connection vanishes, using standard computation (see e. g. \cite[Appendix A]{Gin}), we
get
\[
\bar \nabla _XD_bs=-\frac q{2\left( q-2\right) }\mathrm{Ric}^{\nabla}(X)\cdot s+\frac
1{q-2}X\cdot D_b^2s.  \nonumber
\]

From here
\begin{equation}
D_b^2s=\frac q{4(q-1)}\mathrm{Scal}^\nabla s,  \label{twist_1}
\end{equation}
and we obtain the corresponding upgrading
\begin{equation}
\bar \nabla _XD_bs=\frac q{\left( q-2\right) }\left( \frac {-\mathrm{Ric}^{\nabla}(X)}{2}\cdot s+\frac {1}
{4(q-1)}\mathrm{Scal}^\nabla X\cdot s\right).  \label{twist_2}
\end{equation}

\begin{remark}
The relations (\ref{twist_1}) and (\ref{twist_2}) represent the corresponding version of some classical results concerning
twistor spinors in the classical setting of Riemannian manifolds \cite[Appendix A]{Gin}. Restricting the framework to Riemannian foliations with basic mean curvature, it is also possible to obtain the corresponding results written using only the spin connection $\nabla$ (associated to the so called $W-$ twistor spinors) see \cite[Proposition 3.4.]{Ju2}.
\end{remark}

In the last part of the paper we prove an interesting
property of the zeros of twistor spinors.

First of all let us consider the bundle $E:=S\bigoplus S$, endowed with the
connection
\[
\nabla _X^E\left(
\begin{array}{c}
s_1 \\
s_2
\end{array}
\right) :=\left(
\begin{array}{ccc}
\bar \nabla _{\pi _Q(X)} & \,\, & \frac 1q\pi _Q(X)\cdot  \\
& \,\, &  \\
\frac q{\left( q-2\right) }\left( \mathrm{Ric}^\nabla (X)+\frac 4{q-1}%
\mathrm{Scal}^\nabla X\right) \cdot  & \,\, & \bar \nabla _{\pi _Q(X)}
\end{array}
\right) \left(
\begin{array}{c}
s_1 \\
s_2
\end{array}
\right) ,
\]
for any $s_{1}$, $s_2\in \Gamma (S)$,  as a generalization of the
basic-harmonic case \cite{Ju2} (for the classical case see e.g. \cite{Ba-Fr}).
It is easy to see that if $s$ is a twistor spinor, then the smooth
section $\left(
\begin{array}{c}
s \\
D_bs
\end{array}
\right) $ of $E$ is in fact parallel with respect to $\nabla ^E$, as a
consequence of (\ref{twist_eq}), (\ref{twist_2}) and the definition of $\nabla ^E$.
Considering arguments similar to \cite{Ba-Fr} for the transverse
directions, as the spinors $s$ and $D_bs$ are basic spinors, parallel along
the leaves, defined by (\ref{basic_spinor}), we see that if the manifold $M$ is connected and at a point $x\in M$
we have $s_x=\left(D_bs\right)_x=0$, then $s\equiv 0$ all over the compact manifold
$M$.

For any basic function $f$ we define the \emph{basic Hessian}
associated to the connection $\nabla $
\begin{equation}
\mathrm{Hess}^\nabla (f)(X,Y):=X\left( Y\left( f\right) \right) -\nabla _XY\left(
f\right)   \label{basic_hesse}
\end{equation}
for any $X$, $Y\in \Gamma_b (Q)$.

\begin{remark}
As above, if $T$ is a local transverse manifold, as basic geometric objects in our
framework can be locally projected on $T$, it is easy to see that $\mathrm{Hess}^\nabla $ is just the standard Hessian on the transverse manifold.
\end{remark}

We are now able to prove the corresponding version of a classical property of the zeros of a twistor
spinor \cite{Ba-Fr,Gin}.

\begin{proposition}
On a connected Riemannian foliation of arbitrary codimension $q$ endowed with a nontrivial transversal $\tau$-twistor spinor $s$, the leaves where $s$ vanishes are isolated on the quotient set $M_{/\mathcal{F}}$.
\end{proposition}

\begin{proof}Let us assume $s$ is a nontrivial basic twistor spinor such that $s_x=0$
at $x\in M$. In the following we calculate $\mathrm{Hess}_{x}^\nabla (\left| s\right| ^2)$.

We investigate the two components of the basic Hessian defined by (\ref{basic_hesse}).
\begin{eqnarray*}
X\left( Y\left( \left| s\right| ^2\right) \right)  &=&X\left( 2\,\rm{Re}\left(
\nabla _Ys\mid s\right) \right)  \\
&=&2\,\rm{Re}\left( \nabla _X\nabla _Ys\mid s\right) +2\,\rm{Re}\left( \nabla _Ys\mid \nabla _Xs\right) ,
\end{eqnarray*}
\[
\nabla _XY\left( \left| s\right| ^2\right) =2\,Re\left( \nabla _{\nabla_XY}s\mid s\right) .
\]

As $s_x=0$, the only term we need to study is $\rm{Re}\left( \nabla _Ys\mid \nabla _Xs\right)$.

We now apply (\ref{twist_eq}) and (\ref{modified connection}) and obtain that

\begin{eqnarray*}
\left( \nabla _Ys\mid \nabla _Xs\right)  &=&\left( \bar \nabla _Ys\mid \bar
\nabla _Xs\right) +\frac 12\left\langle Y,\tau \right\rangle \left( s\mid
\nabla _Xs\right)  \\
&&+\frac 12\left\langle X,\tau \right\rangle \left( \nabla _Ys\mid s\right)
+\frac 14\left\langle Y,\tau \right\rangle \left\langle X,\tau \right\rangle \left| s\right| ^2
\\
&=&\frac 1{q^2}\left\langle X\cdot D_bs,Y\cdot D_bs\right\rangle +\frac
12\left\langle Y,\tau \right\rangle \left( s\mid \nabla _Xs\right)  \\
&&+\frac 12\left\langle X,\tau \right\rangle \left( \nabla _Ys\mid s\right)
+\frac 14\left\langle Y,\tau \right\rangle \left\langle X,\tau \right\rangle
\left| s\right| ^2.
\end{eqnarray*}
From \cite[p. 15]{Ba-Fr} we get furthermore that
\[
\left\langle X\cdot D_bs,Y\cdot D_bs\right\rangle =\left\langle
X,Y\right\rangle \left| D_bs\right| ^2,
\]
so at the point $x$,  as $s_x=0$, we obtain
\begin{eqnarray*}
\mathrm{Hess}_{x}^\nabla (\left| s\right| ^2)(X_x,Y_x)&=&2Re\left( \left(\nabla _Ys\right)_x\mid \left(\nabla _Xs\right)_x\right) \\
&=&\frac 2{q^2}\left\langle X_x, Y_x\right\rangle \left| D_bs\right|_x ^2.
\end{eqnarray*}
As $\left( D_bs\right) _x\neq0$ (otherwise, in accordance with the above considerations the twistor spinor $s$ would vanish
everywhere), we get that the basic Hessian of the local function obtained on the local transverse manifold $T$ by projecting the basic
function $\left| s \right| ^2$ is positive defined at $x$; as the basic functions are constant along the leaves of the foliation, the conclusion follows.
\end{proof}

\section{Some physical considerations}

The  Dirac operators in the presence of Riemannian foliations have
attracted much attention in physics. We do not intend to give
any kind of extensive introduction in spin geometry, but some motivating
examples that may lead to possible applications are of interest.

In the last time Sasakian manifolds, as an odd-dimensional cousin of
K\"{a}hler manifolds, have become of high interest in connection with
many modern studies in physics. One of their principal applications in
physics has been in higher-dimensional supergravity, string theory and
$M$-theory where they can provide backgrounds for reduction to
lower-dimensional spacetimes. AdS/CFT conjecture \cite{JM} relates
quantum gravity, in certain backgrounds, to ordinary quantum field
theory without gravity. In particular the AdS/CFT correspondence
relates Sasaki-Einstein geometry, in dimensions five and seven, to
superconformal field theory in dimensions four and three, respectively.

The foliation generated by the Reeb vector field $\xi$ has a transverse
K\"{a}hler structure. If the orbits of $\xi$ are closed, the Sasakian
structure is called {\it quasi-regular}. The Reeb field generates a
locally free $S^1$-action such that the leaf space is an orbifold and
the transverse K\"{a}hler structure projects to it. There are examples
of Sasakian structures which are not quasi-regular \cite{GMSW}. In the
opposite case, if the orbits of $\xi$ do not all close, the Sasakian
structure is said to be {\it irregular}.

The properties of Sasaki-Einstein spaces can be obtained from an
alternative definition of Sasaki-Einstein manifold connected with the
existence of a Killing spinor \cite{CB}. The geometric features of
generic supergravity solutions with unbroken supersymmetry and fluxes,
so the relation between Killing spinor and geometry that admits such a
spinor needs to be further elucidated \cite{NK}. On the other side,
as pointed out in \cite{Gin}, the Killing spinors are highly relevant for the
investigation of supersymmetric models for string theory in dimension 10. From this
point of view we hope that our results concerning transversal Killing spinors would be
helpful for the investigation of this geometrical objects in the particular framework represented
by Riemannian foliations.

Another interesting example is represented by the Euclidean
Taub-Newman-Unti-Tamburino (Taub-NUT) space which appears in various
problems. Hawking \cite{SWH} has suggested that the Euclidean Taub-NUT
metric might give rise to the gravitational analogue of the Yang-Mills
instantons. Also this metric is the space part of the line element of
the Kaluza-Klein monopole.

Iwai and Katayama \cite{IK} generalized the Taub-NUT metrics in the
following way. Let us consider a metric $\bar{g}$ on an open interval
$U$ in $(0, +\infty)$ and a family of Berger metrics $\hat{g}(r)$ on
$S^3$ indexed by $U$. Then the twisted product metrics
$g = \bar{g} +\hat{g}(r)$ on the annulus $U \times S^3 \subset
\mathbb{R}^4 \verb+\+ \{0\}$ is called a generalized Taub-NUT metric.

The Taub-NUT metrics has been drawing wide interest. In particular,
from the viewpoint of dynamical systems, the symmetry of the dynamical
system associated with that metric is similar to that for the
Coulomb/Kepler problem. The four-dimensional problem is reduced to an
$S^1$ action when the associated momentum mapping takes nonzero fixed
values. If the original Hamiltonian system admits a symmetry group that
is commutative with the group used for the reduction, the reduced
Hamiltonian system admits the same symmetry group \cite{MW}. In the
case of the Taub-NUT metrics the reduced Hamiltonian system is the
three-dimensional Kepler problem along with a centrifugal potential and
Dirac's monopole field.

The importance of anomalous Ward identities in particle physics is well
known. The anomalous divergence of the axial vector current in a
background gravitational field is directly related to the index
theorem. Namely, the axial anomaly is interpreted as the index of the
chiral Dirac operator. The difference between the number of null
states of positive and of negative chirality on a ball or annular
domain, may become nonzero for suitable choices of the parameters of
the metric and of the domain when one imposes the Atiyah-Patodi-Singer
spectral condition at the boundary. In the case of the standard
Taub-NUT space, which is hyperK\"{a}hler and therefore scalar-flat,
it can be proved that there are no harmonic $L^2$ spinors using the Lichnerowicz identity and
the infiniteness of the volume \cite{MV}. Moreover there do not exist
$L^2$ harmonic spinors on $\mathbb{R}^4$ for the generalized Taub-NUT
metrics. In particular, the $L^2$ index of the Dirac operator vanishes
\cite{MM}.

\section*{Acknowledgments}
M. Visinescu was supported by CNCS-UEFISCDI, project number
PN-II-ID-PCE-2011-3-0137. G.E. V\^{\i}lcu was supported by
CNCS-UEFISCDI, project number
PN-II-ID-PCE-2011-3-0118.

Adrian Mihai IONESCU\\
Politehnica University of Bucharest,\\
Department of Mathematics,\\
Splaiul Independen\c{t}ei,
Nr. 313, Sector 6,\\
Bucure\c sti 060042-ROMANIA\\
E-mail: adrian{\_}ionescu@mathem.pub.ro\\

Vladimir SLESAR\\
University of Craiova,\\
Department of Mathematics,\\
Str. Al.I. Cuza, Nr. 13,\\
Craiova 200585-ROMANIA\\
E-mail: vlslesar@central.ucv.ro\\

Mihai VISINESCU\\
National Institute for Physics and Nuclear Engineering,\\
Department of Theoretical Physics,\\
P.O.Box M.G.-6, Magurele, Bucharest-ROMANIA\\
E-mail: mvisin@theory.nipne.ro\\

Gabriel Eduard V\^{I}LCU$^{1,2}$ \\
      $^1$University of Bucharest, Faculty of Mathematics and Computer Science,\\
      Research Center in Geometry, Topology and Algebra,\\
      Str. Academiei, Nr. 14, Sector 1,\\
      Bucure\c sti 70109-ROMANIA\\
      E-mail: gvilcu@gta.math.unibuc.ro\\
      $^2$Petroleum-Gas University of Ploie\c sti,\\
       Department of Mathematical Modelling, Economic Analysis and Statistics,\\
      Bulevardul Bucure\c sti, Nr. 39\\
      Ploie\c sti 100680-ROMANIA\\
      E-mail: gvilcu@upg-ploiesti.ro


\begin{thebibliography}{10}



\bibitem{Re}  B. Reinhart, Foliated manifolds with bundle-like metrics.
{\it Ann. Math.} {\bf 69} (1959) 119--132.

\bibitem{To}  Ph. Tondeur, {\it Geometry of Foliations}, (Birkh\"auser, Basel, Boston, 1997).

\bibitem{BJCF}
A. Bejancu, H. R. Farran, {\it Foliations and geometric structures,
Mathematics and Its Applications} (Springer, 2006).

\bibitem{CP}
B.Y. Chen, P. Piccinni, The canonical foliations of a locally conformal K\"{a}hler manifold,
{\it Ann. Mat. Pura Appl.} {\bf 141} (4) (1985) 289--305.

\bibitem{Ko}  Y.~A. Kordyukov, Vanishing theorem for transverse Dirac
operators on Riemannian foliations, {\it Ann. Glob. Anal. Geom.} {\bf 34} (2008)195--211.

\bibitem{Vi1} G.~E. V\^{\i}lcu, Ruled CR-submanifolds of locally conformal K\"{a}hler manifolds, {\it J. Geom. Phys.}
 {\bf 62} (2012)  1366--1372.

\bibitem{Vi2} G.~E. V\^{\i}lcu, Canonical foliations on paraquaternionic Cauchy-Riemann submanifolds, {\it J. Math. Anal. Appl.}
 {\bf 399} (2013) 551--558.

\bibitem{Nic}  L. Nicolaescu, Geometric connections and geometric Dirac
operators on contact manifolds, {\it Diff. Geom. App.}  {\bf 22} (2005) 355--378.

\bibitem{Gl-Ka}  J.~F. Glazebrook, F.~W. Kamber, Transversal Dirac
families in Riemannian foliations, {\it Commun. Math. Phys.}  {\bf 140} (1991) 217--240.

\bibitem{MSS}
G. Marmo, E. J. Saletan, A. Simoni,
A general setting for reduction of dynamical systems,
{\it J. Math. Phys.} {\bf 20} (1979) 856--860.

\bibitem{EW}
E. Witten,
Search for a realistic Kaluza-Klein theory,
{\it Nucl. Phys. B} {\bf 186} (1981) 412--428.

\bibitem{CLC}
W. K. Chan, K. F. Lai, R. Castillo,
Riemannian foliation in $N=1, D=11$ supergravity,
{\it Nuovo Cimento}, {\bf 108} B (1993) 739--752.

\bibitem{Ba-Fr}
H. Baum, T. Friedrich, R. Grunewald, I. Kath, {\it Twistor and Killing spinors
on Riemannian manifolds}, (Teubner-
Verlag, Stuttgart, Leipzig, 1991).

\bibitem{Fr}  T. Friedrich, Der erste Eigenwert des Dirac operators einer
kompakten, Riemannschen Mannigfaltigkeit nichtnegative skalarkr\"ummung,
{\it Math. Nachr.} {\bf 97} (1980) 117--146.

\bibitem{Hij}  O. Hijazi, A conformal lower bound for the smallest
eigenvalue of the Dirac operator and Killing spinors, {\it Commun. Math. Phys.}
{\bf 104} (1986) 151--162.

\bibitem{Gin} N. Ginoux, {\it The Dirac spectrum}, (Springer, 2009).

\bibitem{Al-Gr-Iv} B. Alexandrov, G. Grantcharov, S. Ivanov, An estimate for the first eigenvalue of the Dirac operator on
compact Riemannian spin manifold admitting parallel one-form, {\it J. Geom. Phys.} {\bf 28} (1998) 263--270.

\bibitem{Ju1}  S.~D. Jung, The first eigenvalue of the transversal Dirac
operator. {\it J. Geom. Phys.} {\bf 39} (2001) 253--264.

\bibitem{Ju2}  S.~D. Jung, Y.~B. Moon, The properties of the transversal Killing spinor and transversal twistor spinor
for Riemannian foliations, {\it J. Korean Math. Soc.} {\bf 42}, 6 (2005) 1169--1186.

\bibitem{Gin-Hab1}N. Ginoux, G. Habib, Remarques sur les spineurs de Killing transversaux,
 {\it C. R. Math. Acad. Sci.} {\bf 346} (2008) 657--659.

\bibitem{Gin-Hab2}N. Ginoux, G. Habib, Geometric aspects of transversal Killing spinors
on Riemannian fows, {\it Abh. Math. Semin. Univ. Hambg.} {\bf 78} (2008)  69--90.

\bibitem{Al}  J.~A. \'Alvarez L\'opez, The basic component of the mean
curvature of Riemannian foliations, {\it Ann. Global Anal. Geom.} {\bf 10} (1992) 179--194.

\bibitem{Hab-Ric}  G. Habib, K. Richardson, A brief note on the spectrum
of the basic Dirac operator, {\it Bull. London Math. Soc.} {\bf 41} (2009) 683--690.

\bibitem{Roe}  J. Roe, {\it Elliptic operators, Topology and Asymptotic
Methods}, (CRC Press, 1999).

\bibitem{Va} S. Vacaru, Spinor structures and nonlinear connections in
vector bundles, generalized Lagrange and Finsler spaces. {\it J. Math. Phys.} {\bf 37} (1996) 508--523.

\bibitem{Sl}  V. Slesar, On the Dirac spectrum of Riemannian foliations admitting a basic parallel 1-form,
{\it J. Geom. Phys.} {\bf 62} (2012) 804--813.

\bibitem{Do}  D. Dom\'\i nguez, A tenseness theorem for Riemannian
foliations, {\it C. R. Acad. Sci.} {\bf 320} (1995) 1331--1335.

\bibitem{Mas}  A. Mason, An application of stochastic flows to
Riemannian foliations, {\it Houston J. Math.} {\bf 26} (2000) 481--515.

\bibitem{JM}
J. M. Maldacena,
The large $N$ limit of superconformal field theories and supergravity,
{\it Adv. Theor. Math. Phys.} {\bf 2} (1998) 231--252.

\bibitem{GMSW}
J. P. Gauntlett, D. Martelli, J. Sparks, D. Waldram,
Sasaki-Einstein metrics on $S^2 \times S^3$,
{\it Adv. Theor. Math. Phys.} {\bf 8} (2004) 711--734.

\bibitem{CB}
C. B\"{a}r,
Real Killing spinors and holonomy,
{\it Commun. Math. Phys.} {\bf 154} (1993) 509--521.

\bibitem{NK}
N. Kim,
Sasaki-Einstein manifolds and their spinorial geometry,
{\it J. Korean Phys. Soc.} {\bf 48} (2006) 197--201.

\bibitem{SWH}
S. W. Hawking,
Gravitational instantons,
{\it Phys. Lett. A} {\bf 60} (1977) 81--83.

\bibitem{IK}
T. Iwai, N. Katayama,
On extended Taub-NUT metrics,
{\it J. Geom. Phys.} {\bf 12} (1993) 55--75.

\bibitem{MW}
J. E. Marsden, A. Weinstein,
Reduction of symplectic manifolds with symmetries,
{\it Rep. Math. Phys.} {\bf 5} (1974) 121--130.

\bibitem{MV}
S. Moroianu, M. Visinescu,
$L^2$-index of the Dirac operator of generalized Euclidean Taub-NUT
metrics,
{\it J. Phys. A -Math. Gen.} {\bf 39} (2006) 6575--6581.

\bibitem{MM}
A. Moroianu, S. Moroianu,
The Dirac operator on generalized Taub-NUT spaces,
{\it Commun. Math. Phys.} {\bf 305} (2011) 641--656.
\end{thebibliography}
\end{document}